\documentclass[11pt]{article}

\usepackage{amsthm}
\usepackage{graphicx} % support the \includegraphics command and options
\usepackage{array} % for better arrays (eg matrices) in maths

\usepackage{amsmath, amssymb, amsfonts, verbatim}
\usepackage{hyphenat,epsfig,subcaption,multirow}

\usepackage[usenames,dvipsnames]{xcolor}
\usepackage[lined,boxed,ruled,norelsize,algo2e,linesnumbered]{algorithm2e}

\usepackage{tcolorbox}
\tcbuselibrary{skins,breakable}
\tcbset{enhanced jigsaw}

\usepackage[compact]{titlesec}

\definecolor{DarkRed}{rgb}{0.5,0.1,0.1}
\definecolor{DarkBlue}{rgb}{0.1,0.1,0.5}

\usepackage{nameref}
\definecolor{ForestGreen}{rgb}{0.1333,0.5451,0.1333}
\definecolor{Red}{rgb}{0.9,0,0}
\usepackage[linktocpage=true,
	pagebackref=true,colorlinks,
	linkcolor=DarkRed,citecolor=ForestGreen,
	bookmarks,bookmarksopen,bookmarksnumbered]
	{hyperref}

\usepackage{bm}
\usepackage{url}
\usepackage{xspace}
\usepackage[mathscr]{euscript}

\usepackage{tikz}
\usetikzlibrary{arrows}
\usetikzlibrary{arrows.meta}
\usetikzlibrary{shapes}
\usetikzlibrary{backgrounds}
\usetikzlibrary{positioning}
\usetikzlibrary{decorations.markings}
\usetikzlibrary{patterns}
\usetikzlibrary{calc}
\usetikzlibrary{fit}
\usetikzlibrary{snakes}

\usepackage{mdframed}

\usepackage[noend]{algpseudocode}
\makeatletter
\def\BState{\State\hskip-\ALG@thistlm}
\makeatother

\usepackage{cite}
\usepackage{enumerate}
\usepackage[shortlabels]{enumitem}

\usepackage[margin=1in]{geometry}

\newtheorem{theorem}{Theorem}
\newtheorem{lemma}{Lemma}[section]

\newtheorem{claim}[lemma]{Claim}

\newtheorem{definition}[lemma]{Definition}

\newtheorem{problem}{Problem}
\newtheorem{remark}[lemma]{Remark}

\newtheorem*{claim*}{Claim}
\newtheorem*{proposition*}{Proposition}
\newtheorem*{lemma*}{Lemma}
\newtheorem*{problem*}{Problem}

\newtheorem{mdresult}{Result}
\newenvironment{result}{\begin{mdframed}[backgroundcolor=lightgray!40,topline=false,rightline=false,leftline=false,bottomline=false,innertopmargin=2pt]\begin{mdresult}}{\end{mdresult}\end{mdframed}}

\newtheorem{mdinvariant}{Invariant}

\allowdisplaybreaks

\renewcommand{\qed}{\nobreak \ifvmode \relax \else
      \ifdim\lastskip<1.5em \hskip-\lastskip
      \hskip1.5em plus0em minus0.5em \fi \nobreak
      \vrule height0.75em width0.5em depth0.25em\fi}

\newcommand{\toShrink}{-.20cm}
\newcommand{\toShrinkEnu}{-.2cm}

\newcommand{\Ot}{\ensuremath{\widetilde{O}}}
\newcommand{\eps}{\ensuremath{\varepsilon}}
\newcommand{\Paren}[1]{\Big(#1\Big)}

\newcommand{\bracket}[1]{\left[#1\right]}
\newcommand{\paren}[1]{\ensuremath{\left(#1\right)}\xspace}
\newcommand{\card}[1]{\left\vert{#1}\right\vert}

\newcommand{\IR}{\ensuremath{\mathbb{R}}}

\newcommand{\IN}{\ensuremath{\mathbb{N}}}

\newcommand{\expect}[1]{\Exp\bracket{#1}}

\newcommand{\set}[1]{\ensuremath{\left\{ #1 \right\}}}
\newcommand{\poly}{\mbox{\rm poly}}
\newcommand{\polylog}{\mbox{\rm  polylog}}

\newcommand{\alg}{\ensuremath{\mathcal{A}}\xspace}

\DeclareMathOperator*{\Exp}{\ensuremath{{\mathbb{E}}}}
\DeclareMathOperator*{\Prob}{\ensuremath{\textbf{Pr}}}
\renewcommand{\Pr}{\Prob}

\newcommand{\Ex}{\Exp}

\newcommand{\etal}{{ et al.\,}}

\newcommand{\FG}{\ensuremath{\mathcal{G}}}

\newenvironment{tbox}{\begin{tcolorbox}[
		enlarge top by=5pt,
		enlarge bottom by=5pt,
		 breakable,
		 boxsep=0pt,
                  left=4pt,
                  right=4pt,
                  top=10pt,
                  arc=0pt,
                  boxrule=1pt,toprule=1pt,
                  colback=white
                  ]%%
	}
{\end{tcolorbox}}

\newcommand{\bM}{\bm{M}}

\title{When Algorithms for Maximal Independent Set and \\ Maximal Matching Run in Sublinear-Time\footnote{A preliminary version of this paper appeared in the proceedings of ICALP'19.}}
\author{Sepehr Assadi\thanks{Department of Computer Science, Rutgers University. Email: {{\small {\tt sepehr.assadi@rutgers.edu}}}}
\and Shay Solomon\thanks{School of Electrical Engineering, Tel Aviv University. Email: {\small \texttt{shayso@post.tau.ac.il.}}}}

\date{}

\begin{document}
\maketitle

%\thispagestyle{empty}
%\input{abstractShay}
%\setcounter{page}{0}
%\clearpage

\thispagestyle{empty}
\begin{abstract}

Maximal independent set (MIS), maximal matching (MM), and $(\Delta+1)$-coloring in graphs of maximum degree $\Delta$ are among the most prominent algorithmic graph theory problems. They are all solvable by a simple linear-time greedy algorithm and up until very recently this constituted the state-of-the-art. In SODA 2019, Assadi, Chen, and Khanna gave a randomized algorithm for $(\Delta+1)$-coloring
that runs in $\Ot(n\sqrt{n})$ time\footnote{Here, and throughout the paper, we define $\Ot(f(n)) := O(f(n) \cdot \polylog{(n)})$ to suppress log-factors.}, which even for moderately dense graphs is sublinear in the input size.
The work of Assadi et al. however contained a spoiler for MIS and MM: neither problems provably admits a sublinear-time algorithm in general graphs.
In this work, we dig deeper into the possibility of achieving sublinear-time algorithms for MIS and MM.

\smallskip

	The neighborhood independence number of a graph $G$, denoted by $\beta(G)$, is the size of the largest independent set in the neighborhood of any vertex.
	We identify $\beta(G)$ as the ``right'' parameter to measure the runtime of MIS and MM algorithms:
	Although graphs of bounded neighborhood independence may be very dense (clique is one example), we prove that carefully chosen
	variants of greedy algorithms for MIS and MM run in $O(n\beta(G))$ and $O(n\log{n}\cdot\beta(G))$ time respectively on any $n$-vertex graph $G$.
	We complement this positive result by observing that a simple extension of the lower bound of Assadi\etal implies that $\Omega(n\beta(G))$ time is also necessary for any algorithm to either problem for all values of $\beta(G)$  from $1$ to $\Theta(n)$.
	We note that our algorithm for MIS is deterministic while for MM we use randomization which we prove is unavoidable: any deterministic algorithm for MM requires $\Omega(n^2)$ time even for $\beta(G) = 2$.
	
	\smallskip
	
	Graphs with bounded neighborhood independence, already for constant $\beta = \beta(G)$, constitute a rich family of possibly dense graphs, including line graphs, proper interval graphs, unit-disk graphs, claw-free graphs, and graphs of bounded growth.
	Our results  suggest that even though MIS and MM do not admit sublinear-time algorithms in general graphs,
	one can still solve both problems in sublinear time for a wide range of $\beta(G) \ll n$.

	\smallskip
		
	Finally, by observing that the lower bound of $\Omega(n\sqrt{n})$ time for $(\Delta+1)$-coloring due to Assadi\etal applies to graphs of (small) constant neighborhood independence, we unveil an intriguing  separation between the time complexity of MIS and MM, and that of $(\Delta+1)$-coloring:  while the time complexity of MIS and MM is strictly higher than that of $(\Delta+1)$ coloring in general graphs, the exact opposite relation holds for graphs with small neighborhood independence.

\end{abstract}

\setcounter{page}{0}
\clearpage

%%\thispagestyle{empty}
%%\setcounter{tocdepth}{3}
%%\tableofcontents
%%\clearpage
%%\setcounter{page}{1}

\renewcommand{\deg}[1]{\ensuremath{\textnormal{\textsf{deg}}(#1)}}
\newcommand{\degto}[2]{\ensuremath{\textnormal{\textsf{deg}}_{#2}(#1)}}

\newcommand{\mis}{\ensuremath{\mathcal{M}}}

\newcommand{\markn}{\ensuremath{\textnormal{\texttt{mark}}}}

\newcommand{\false}{\textnormal{FALSE}}
\newcommand{\true}{\textnormal{TRUE}}

\section{Introduction}\label{sec:intro}

Maximal independent set (MIS) and maximal matching (MM) are two of the most prominent graph problems with a wide range of applications in particular to {symmetry breaking}. Algorithmic study of these problems can be traced back to at least four decades ago
in the pioneering work of~\cite{KarpW84,Luby85,AlonBI86,IsraelI86} on PRAM algorithms. These problems have since been studied extensively in various  models
 including distributed algorithms~\cite{Linial92,HanckowiakKP98,KuhnMW04,LenzenW11,BarenboimEPS12,Ghaffari16,FischerGK17}, dynamic
 algorithms~\cite{NeimanS13,BaswanaGS15,Solomon16,OnakSSW18,AssadiOSS18,AssadiOSS19},
streaming algorithms~\cite{FKMSZ05,HalldorssonHLS16,CormodeDK17,AssadiCK19}, massively parallel computation (MPC) algorithms~\cite{LattanziMSV11,GhaffariGMR18,BehnezhadHH19}, local computation
algorithms (LCA)~\cite{RubinfeldTVX11,AlonRVX12,Ghaffari16,FraigniaudHK16,LeviRY17,GhaffariU19}, and numerous others.

In this paper, we consider the \textbf{time complexity} of MIS and MM (in the centralized setting) and focus on one of the most basic questions regarding these two problems:
\vspace{-2pt}
\begin{quote}
\emph{How fast can we solve maximal independent set and maximal matching problems?}
\end{quote}
\vspace{-2pt}

At first glance, the answer to this question may sound obvious: there are text-book greedy algorithms for both problems that run in linear time and ``of course'' one cannot solve these problems faster as just reading the input takes linear time.
This answer however is not quite warranted: for the closely related problem of $(\Delta+1)$-(vertex) coloring, very recently Assadi, Chen, and Khanna~\cite{AssadiCK19} gave a randomized
algorithm that runs in only $\Ot(n\sqrt{n})$ time on any $n$-vertex graph with high probability\footnote{We say an event happens with high probability if it happens with probability at least $1-1/\poly{(n)}$.}. This means that
even for moderately dense graphs, one can indeed color the graph faster than reading the entire input, i.e., in sublinear time.

The Assadi-Chen-Khanna algorithm hints that one could perhaps hope for sublinear-time algorithms for MIS and MM as well. Unfortunately however,
the work of~\cite{AssadiCK19} already contained a spoiler: neither MIS nor MM provably admits a sublinear-time algorithm in general graphs.

In this work, we show that despite the negative result of~\cite{AssadiCK19} for MIS and MM, the hope for obtaining sublinear-time algorithms for these problems need not  be  short lived. In particular, we
identify a key parameter of the graph, namely the \emph{neighborhood independence number}, that provides a more nuanced measure of runtime for these problems 
and show that both problems can be solved much faster when neighborhood independence is small. This in turn
gives rise to sublinear-time algorithms for MIS and MM on a rich family of graphs with bounded neighborhood independence. In the following, we elaborate more on our results.

\subsection{Our Contributions}\label{sec:results}

For a graph $G(V,E)$, the neighborhood independence number of $G$, denoted by $\beta(G)$, is defined as the size of the largest independent set in the graph in which all vertices of the independent set are incident
on some shared vertex $v \in V$. Our main result  is as follows:

\begin{result}\label{res:main}
	There exist algorithms that given a graph $G(V,E)$  find $(i)$ a maximal independent set of $G$ deterministically in $O(n \cdot \beta(G))$ time, and $(ii)$
	a maximal matching of $G$ randomly in $O(n\log{n} \cdot \beta(G))$ time in expectation and with high probability.
\end{result}

When considering sublinear-time algorithms, specifying the exact data model is important as the algorithm cannot even read the entire input once. We assume that the input graph is presented in the adjacency array representation, i.e., for each vertex $v \in V$, we are given degree $\deg{v}$ of $v$ followed by an array of length $\deg{v}$ containing all neighbors of $v$ in arbitrary order.
This way, we can access the degree of any vertex $v$ or its $i$-th neighbor for $i \in [\deg{v}]$ in $O(1)$ time. We also make the common assumption that a random number from $1$ to $n$ can be generated in $O(1)$ time.
This is a standard input representation for graph problems and is commonly used in the area of sublinear-time algorithms (see, e.g.~\cite{GoelKK09,GoelKK10,OnakRRR12}).
Let us now elaborate on several aspects of Result~\ref{res:main}.

\paragraph{Optimality of Our Bounds.} Assadi\etal~\cite{AssadiCK19} proved that any algorithm
for MIS or MM requires $\Omega(n^2)$ time in general.
These lower bounds can be extended in an easy way to prove that $\Omega(n \cdot \beta)$ time is also necessary for both problems on graphs with neighborhood independence $\beta(G)=\beta$.
Indeed, independently sample $t:=n/\beta$ graphs $G_1,\ldots,G_t$ each on $\beta$ vertices from the hard distribution of graphs in~\cite{AssadiCK19} and let $G$ be the union of these graphs. Clearly, $\beta(G) \leq \beta$ and it follows
that since solving MIS or MM on each graph $G_i$ requires $\Omega(\beta^2)$ time by the lower bound of~\cite{AssadiCK19}, solving $t$ \emph{independent} copies requires $\Omega(t \cdot \beta^2) = \Omega(n\beta)$ time.
As such, our Result~\ref{res:main} is optimal for every $\beta$ ranging from a constant to $\Theta(n)$ (up to a constant factor for MIS and $O(\log{n})$ for MM).

\paragraph{Our Algorithms.} Both our algorithms for MIS and MM in Result~\ref{res:main} are similar to the standard greedy algorithms,
though they require careful adjustments and implementation. Specifically, the algorithm for MIS is the standard deterministic greedy algorithm (with minimal modification)
and for MM we use a careful implementation of the (modified) randomized greedy algorithm (see, e.g.~\cite{DyerF91,AronsonDFS95,MillerP97,PoloczekS12}).
The novelty of our work mainly lies in the analysis of these algorithms.
We show, perhaps surprisingly,  that already-known algorithms can in fact achieve an improved performance and run in sublinear-time for graphs with bounded neighborhood independence \emph{even} when the value of $\beta(G)$ is
unknown to the algorithms. Combined with the optimality of our bounds mentioned earlier, we believe that this makes neighborhood independence number  an ideal parameter for measuring the runtime of MIS and MM algorithms.

\paragraph{Determinism and Randomization.} Our MIS algorithm in Result~\ref{res:main} is  deterministic which  is a rare occurrence in the realm of sublinear-time algorithms.
But for MM, we again fall back on randomization to achieve sublinear-time performance. This is not a coincidence however: we prove in Theorem~\ref{thm:mm-lower} that any deterministic algorithm for MM
requires $\Omega(n^2)$ time {even} on graphs with constant neighborhood independence number. This also suggests a separation in the time complexity of MIS and MM for deterministic algorithms.

\paragraph{Bounded Neighborhood Independence.}  Our Result~\ref{res:main} is particularly interesting for graphs with constant neighborhood independence as we obtain quite fast algorithms
with running time $O(n)$ and $O(n\log{n})$ for MIS and MM, respectively. Graphs with constant neighborhood independence  capture a rich family of graphs;
several illustrative examples are as follows:

\begin{itemize}
	\item \emph{Line graphs:} For any arbitrary graph $G$, the neighborhood independence number of its line graph $L(G)$ is at most $2$. More generally, for any $r$-hyper graph $\mathcal{H}$ in which each hyper-edge connects at most $r$ vertices,
	$\beta(L(\mathcal{H})) \leq r$.
		
		\smallskip
		
	\item \emph{Bounded-growth graphs:} A graph $G(V,E)$ is said to be of bounded growth iff there exists a function such that for every vertex $v \in V$ and integer $r \geq 1$, the size of the largest independent set in the $r$-neighborhood
	of $v$ is bounded by $f(r)$. Bounded-growth graphs in turn capture several intersection graphs of geometrical objects such
	 as proper interval graphs~\cite{HalldorssonKS03}, unit-disk graphs~\cite{Halldorsson09}, quasi-unit-disk graphs~\cite{KuhnWZ08}, and general disc graphs~\cite{HalldorssonK15}.
	
	\smallskip
	
	\item \emph{Claw-free graphs:} Graphs with neighborhood independence $\beta$ can be alternatively defined as $\beta$-claw-free graphs, i.e., graphs that do not contain $K_{1,\beta}$ as an induced subgraph.
	 Claw-free graphs have been subject of extensive study in structural graph theory; see the series of papers by Chudnovsky and Seymour, starting with~\cite{ChudnovskyS05}, and
	the survey by Faudree\etal~\cite{FaudreeFR97}.
	%~\cite{Minty80}
\end{itemize}

Above graphs appear naturally in the context of symmetry breaking problems (for instance in the study of wireless networks), and there have been numerous works on MIS and MM in graphs with bounded neighborhood independence
and their special cases (see, e.g.~\cite{KuhnWZ08,SchneiderW08,Halldorsson09,BarenboimE11,BarenboimE13,HalldorssonK15,FischerGK17,BarenboimO20a,BarenboimO20b} and references therein).
%Our results in this paper makes further progress in understanding the (time-)complexity of MIS and MM in these graphs.

\subsection{Other Implications}\label{sec:implications}

Despite the simplicity of our algorithms in Result~\ref{res:main}, they lead to several interesting implications, when combined with some known results and/or techniques:

\begin{enumerate}[(a)]
	\item \emph{Approximate vertex cover and matching:} Our MM algorithm in Result~\ref{res:main} combined with well-known properties of maximal matchings
	implies an $O(n\log{n} \cdot \beta(G))$ time $2$-approximation algorithm to both maximum matching and minimum vertex cover. For graphs with constant neighborhood independence, our results
	improve upon the sublinear-time algorithms of~\cite{OnakRRR12} that achieve $(2+\eps)$-approximation to the \emph{size} of the optimal solution to both problems but do not find the actual edges or vertices
	in $\Ot_{\eps}(n)$ time on general graphs.
	
	\smallskip
	
	\item \emph{Caro-Wei bound and  approximation of maximum independent set:} The Caro-Wei bound~\cite{Caro79,Wei81} states that any graph $G(V,E)$ contains an independent set of size at least $\sum_{v \in V} \frac{1}{\deg{v}+1}$,
	and there is a substantial interest in obtaining independent sets of this size (see, e.g.~\cite{HalldorssonHLS16,HalldorssonK15,CormodeDK17} and references therein). One standard way of obtaining such independent set is to
	run the greedy MIS algorithm on the vertices of the graph in the increasing order of their degrees. As our Result~\ref{res:main} implies that one can implement the greedy MIS algorithm for \emph{any} ordering of vertices,
	we can sort the vertices in $O(n)$ time and then run our deterministic algorithm with this order to obtain an independent set with Caro-Wei bound size in $O(n\beta(G))$ time. Additionally, it is easy
	to see that on graphs with $\beta(G) = \beta$, any MIS is a $\beta$-approximation to the \emph{maximum} independent set (see, e.g.~\cite{KuhnNMW05,SchneiderW08}).
	We hence also obtain a constant factor approximation in $O(n)$ time for maximum independent set on graphs with bounded neighborhood independence.
	
	\smallskip
	
	\item \emph{Separation of $(\Delta+1)$-coloring with MIS and MM:} Assadi\etal~\cite{AssadiCK19} gave an $\Ot(n\sqrt{n})$ time algorithm for $(\Delta+1)$ coloring and an $\Omega(n^2)$ time lower bound for MIS and MM on
	general graphs. It is also shown
	in~\cite{AssadiCK19} that $(\Delta+1)$ coloring requires $\Omega(n\sqrt{n})$ time and in fact the lower bound holds for graphs with constant neighborhood independence. Together with
	our Result~\ref{res:main}, this implies an interesting separation between the time-complexity of MIS and MM, and that of $(\Delta+1)$-coloring: while the time complexity of MIS and MM is strictly higher than
	 that of $(\Delta+1)$ coloring in general graphs, the exact opposite relation holds for graphs with small neighborhood independence number.
	
	 \smallskip
	
	 \item \emph{Efficient MM computation via MIS on line graphs:} The line graph $L(G)$ of a graph $G$  contains $m$ vertices corresponding to edges of $G$ and up to $O(mn)$ edges.
	 Moreover, for any  graph $G$, $\beta(L(G)) \leq 2$. As an MIS in $L(G)$ corresponds to an MM in $G$, our results suggest that despite the larger size of $L(G)$, perhaps surprisingly, computing an MM of $G$
	 through computing an MIS for $L(G)$ is just as efficient as directly computing an MM of $G$ (assuming direct access to $L(G)$). This observation may come into play in real-life situations where there is no direct access
	 to the graph but rather only to its line graph.
	
	\end{enumerate}

\subsection*{Preliminaries and Notation}
For a graph $G(V,E)$ and vertex $v \in V$, $N(v)$ and $\deg{v}$ denote
the neighbor-set and degree of a vertex $v$, respectively. For a subset $U \subseteq V$, $\degto{v}{U}$ denotes the degree of $v$ to vertices in $U$.
Denote by $\beta(G)$ the neighborhood independence number of graph $G$.

\section{Technical and Conceptual Highlights}\label{sec:high-level}

Our first (non-technical) contribution is in identifying the neighborhood independence number as the ``right'' measure of time-complexity for both  MIS and MM.
We then show that surprisingly simple algorithms for these problems
run in sublinear-time on graphs with bounded $\beta(G)$.

The textbook greedy algorithm for MIS works as follows:   scan the vertices in an arbitrary order and add each scanned vertex to a set $\mis$ iff it does not already have a neighbor in $\mis$. Clearly the runtime of this algorithm is $\Theta(\sum_{v \in V} \deg{v}) = \Theta(m)$
and this  bound does not improve for graphs with small $\beta$.
We can slightly tweak this algorithm by making every vertex that joins $\mis$ to \emph{mark} all its neighbors and  simply ignore 
scanning the already marked vertices. This tweak however is not useful in general graphs
as the algorithm may waste time by repeatedly marking the same vertices over and over again without making much further progress
(the complete bipartite graph is an extreme example).
The same problem manifests itself in  other algorithms, including those for MM, and is  at the root of the lower bounds in~\cite{AssadiCK19} for sublinear-time computation of MIS and MM. 

We prove that this issue cannot arise in graphs with bounded neighborhood independence.
Noting that the runtime of the greedy MIS algorithm that uses ``marks'' is $\Theta(m_\mis)$, where we define $m_\mis := \sum_{v \in \mis} \deg{v}$,
a key observation is that $m_\mis$ is much smaller than $m$ when $\beta$ is small.
Indeed, as the vertices of $\mis$ form an independent set, all the edges incident on $\mis$ lead to $V \setminus \mis$, and so if $m_\mis$ is large, then the \emph{average} degree of $V \setminus \mis$ \emph{to} $\mis$ cannot be ``too small''; however,
the latter average degree cannot be larger than $\beta$ as otherwise there is some vertex in $V \setminus \mis$ that is incident to more than $\beta$ independent vertices, a contradiction. This is all we need to conclude that the runtime of the greedy MIS algorithm that uses marks is bounded by $O(n \cdot \beta)$.

Both the MIS algorithm and its analysis are remarkably simple, and {\em in hindsight}, %it makes perfect sense %that things should be fast and easy when $\beta$ is small,
this is not surprising since this parameter $\beta$ is in a sense ``tailored'' to the MIS problem.
Although MM and MIS problems are intimately connected to each other,
the MM problem appears to be much more intricate for graphs with bounded neighborhood independence. Indeed, while the set $U$ of unmatched vertices in any MM forms an independent set 
and hence total number $m_U$ of edges incident on $U$ cannot be too large by the above argument, the runtime of greedy or any other algorithm cannot be bounded in terms of $m_U$ (as $m_U$ can simply be zero). 
In fact, it is provably impossible to adjust our argument for MIS to the MM problem
due to our lower bound for deterministic MM algorithms (Theorem~\ref{thm:mm-lower}) that shows that any such algorithm must incur a runtime of $\Omega(n^2)$ even for $\beta = 2$.

The main technical contribution of this paper is thus in obtaining a fast randomized MM algorithm for graphs with bounded $\beta$. 
Our starting point is the modified randomized greedy (MRG) algorithm of~\cite{DyerF91,AronsonDFS95} that finds an MM by iteratively picking an \emph{unmatched} vertex $u$ uniformly at random and matching it to a uniformly 
at random chosen \emph{unmatched} neighbor $v \in N(u)$. On its own, this standard algorithm does not benefit from small values of $\beta$: while picking an unmatched vertex $u$ is easy, finding an unmatched neighbor $v$ for $u$ is too 
time-consuming
in general. We instead make the following simple but crucial modification: instead of picking $v$ from unmatched neighbors of $u$, we simply sample $v$ from the set of \emph{all} neighbors of $u$  and only match it to $u$ if it is also unmatched; otherwise
we sample another vertex $u$ and continue like this (additional care is needed to ensure that this process even terminates but we postpone the details to Section~\ref{sec:mm-upper}).

To analyze the runtime of this modified algorithm, we leverage the above argument for MIS and take it to the next step to prove a basic structural property of graphs with bounded neighborhood independence: for any set $P$ of vertices, 
a constant fraction of vertices are such that their inner degree inside $P$ is ``not much smaller'' than their total degree (depending both on $\beta$ and size of $P$). Letting $P$ to be the set of unmatched vertices in the above algorithm
allows us to bound the number of iterations made by the algorithm before finding the next matching edge, and ultimately bounding the overall runtime of the algorithm  by $O(n\log{n} \cdot \beta)$  in expectation and with high probability. 
%The above discussion oversimplifies many details. 

%We remark that the randomized greedy algorithm described here

\subsubsection*{Technical Comparison with Prior Work}
Our work is most closely related to the $\Ot(n\sqrt{n})$-time $(\Delta+1)$-coloring algorithm of Assadi, Chen, and Khanna~\cite{AssadiCK19} (and their $\Omega(n^2)$ time lower bounds for MIS and MM on general graphs), 
as well as the series of work by Goel, Kapralov, and
Khanna~\cite{GoelKK09,GoelKK09b,GoelKK10} on finding perfect matchings in regular bipartite graphs that culminated in an $O(n\log{n})$ time algorithm.

The coloring algorithm of~\cite{AssadiCK19} works by \emph{non-adaptively sparsifying} the graph into $O(n\log^{2}(n))$ edges in $\Ot(n\sqrt{n})$ time in such a way that a $(\Delta+1)$ coloring of the original graph can be found quickly from this sparsifier.
The algorithms in~\cite{GoelKK09,GoelKK09b} were also based on the high-level idea of sparsification but the final work in this series~\cite{GoelKK10} instead used a \emph{(truncated) random walk} approach to speed up augmenting path computations in regular graphs.
The sparsification methods used in~\cite{AssadiCK19,GoelKK09,GoelKK09b} as well as the random walk approach of~\cite{GoelKK10} are all quite different from our techniques in this paper that are tailored to graphs with bounded neighborhood independence.
Moreover, even though every perfect matching is clearly maximal, our results and~\cite{GoelKK09,GoelKK09b,GoelKK10} are incomparable as $d$-regular bipartite graphs and graphs
with bounded neighborhood independence are in a sense the exact opposite of each other: for a $d$-regular bipartite graph, $\beta(G) = d$ which is the largest possible for graphs with maximum degree $d$.

\section{Maximal Independent Set}\label{sec:mis-upper}

The standard greedy algorithm for MIS works as follows: Iterate over vertices of the graph in an arbitrary order and insert each one to an initially empty set $\mis$ if none of its neighbors have already been inserted to $\mis$.
By the time all vertices have been processed, $\mis$ clearly provides an MIS of the input graph. See Algorithm~\ref{alg:mis} for a pseudo-code.

\begin{algorithm2e}[h!]

\caption{ \vspace{2pt} The (Deterministic) Greedy Algorithm for Maximal Independent Set. }\label{alg:mis}
\SetAlgoLined

\smallskip

\textbf{Input:} An $n$-vertex graph $G(V,E)$ given in adjacency array representation.

\smallskip

\textbf{Output:} An MIS $\mis$ of $G$.

\smallskip

Initialize $\mis = \emptyset$ and $\markn[v_i] \leftarrow \false$ for all vertices $v_i \in V$ where $V := \{v_1,\ldots,v_n\}$.

\For{$i=1$ \textbf{to} $n$}{

\smallskip

\If{$\markn[v_i] = \false$}{
add $v_i$ to $\mis$ and set $\markn[u] \leftarrow \true$ for all $u \in N(v_i)$.
}
}

\textbf{Return} $\mis$.

\end{algorithm2e}

We prove that this algorithm is fast on graphs with bounded neighborhood independence. 

\begin{theorem}\label{thm:mis-upper}
	The greedy MIS algorithm (as specified by Algorithm~\ref{alg:mis}) computes a maximal independent set of a graph $G$ given in adjacency array representation  in $O(n \cdot \beta(G))$ time.
\end{theorem}
\begin{proof}
Let $G(V,E)$ be an arbitrary graph. Suppose we run Algorithm~\ref{alg:mis} on $G$ and obtain $\mis$ as the resulting MIS.
To prove Theorem~\ref{thm:mis-upper}, we use the following two simple claims.

\begin{claim}\label{clm:mis-alg-time}
	The time spent by Algorithm~\ref{alg:mis} on a graph $G(V,E)$ is $O(n + \sum_{v \in \mis} \deg{v})$.
\end{claim}
\begin{proof}
	Iterating over vertices in the \textbf{for} loop takes $O(n)$ time. Beyond that, for each vertex joining the MIS $\mis$, we spend time that is linear in its degree to mark all its neighbors. 
\end{proof}

\begin{claim}\label{clm:mis-alg-degree}
	For any independent set $I \subseteq V$ in $G$, $\sum_{v \in I} \deg{v} \leq n \cdot \beta(G)$.
\end{claim}
\begin{proof}
	Let $E(I)$ denote the edges incident on vertices in the independent set $I$. Since $I$ is an independent set, these edges connect vertices of $I$ with vertices of $V \setminus I$. Suppose towards a contradiction that
	$\card{E(I)} = \sum_{v \in I} \deg{v} > n \cdot \beta(G)$. By a double counting argument, there must exist a vertex $v$ in $V \setminus I$ with at least $\card{E(I)}/\card{V \setminus I} > \beta(G)$ neighbors in $I$. But since $I$ is an independent set, this 
	means that there exists an independent set of size $> \beta(G)$ in the neighborhood of $v$, which contradicts the fact that $\beta(G)$ is the neighborhood independence number of $G$. 
\end{proof}

Theorem~\ref{thm:mis-upper} now follows from Claims~\ref{clm:mis-alg-time} and~\ref{clm:mis-alg-degree} as $\mis$ is an independent set of $G$. 
\end{proof}

\section{Maximal Matching}\label{sec:mm-upper}

\newcommand{\Vexp}{\ensuremath{V_{\textnormal{\texttt{exp}}}}}
\newcommand{\Vrem}{\ensuremath{V_{\textnormal{\texttt{rem}}}}}

\newcommand{\trem}{\ensuremath{\tau_{\textnormal{\texttt{rem}}}}}

\newcommand{\misalg}{\ensuremath{\textnormal{\textsf{MIS-Algorithm}}}\xspace}
\newcommand{\missampler}{\ensuremath{\textnormal{\textsf{MIS-Sampler}}}\xspace}

We now consider the maximal matching (MM) problem. Similar to MIS, a standard greedy algorithm for MM is to iterate over the vertices in arbitrary order and match each vertex to one of their unmatched neighbors (if any).
However, as we show in Section~\ref{sec:mm-lower} this and any other \emph{deterministic} algorithm for MM, cannot run in sublinear-time even when $\beta(G) = 2$.

We instead consider the following variant of the greedy algorithm, referred to as the \emph{(modified) randomized greedy algorithm}, put forward by~\cite{DyerF91,AronsonDFS95} and  extensively studied in the literature primarily with respect to its approximation ratio for the maximum matching problem (see, e.g.~\cite{MillerP97,PoloczekS12} and the references therein). Pick an unmatched vertex $u$ uniformly
at random; pick an unmatched vertex $v$ incident on $u$ uniformly at random and add $(u,v)$ to the matching $M$; repeat as long as there is an unmatched edge left in the graph. It is easy to see that at the end of the algorithm $M$ will be an MM of $G$.

As it is, this algorithm is not suitable for our purpose as finding an unmatched vertex $v$ incident on $u$ is too costly. We thus instead consider the following variant which samples the set $v$ from all neighbors of $u$ and 
only match it to $u$ if $v$ is also unmatched (we also change the final check of the algorithm for maximality of $M$ with a faster computation). See Algorithm~\ref{alg:mm} for a pseudo-code after proper modifications.

\begin{algorithm2e}[h!]

\caption{ \vspace{2pt} The (Modified) Randomized Greedy Algorithm for Maximal Matching. }\label{alg:mm}
\SetAlgoLined

\smallskip

\textbf{Input:} An $n$-vertex graph $G(V,E)$ given in adjacency array representation.

\smallskip

\textbf{Output:} A maximal matching $M$ of $G$.

\smallskip

Initialize $M = \emptyset$ and $U = V$. %and define the threshold $\tau := 4\beta(G)\log{n}$.

\smallskip

\While{$U \neq \emptyset$}{

\smallskip

 Define the threshold $\tau:= \tau(U) = \frac{4n \cdot \beta(G)}{\card{U}}$.

\smallskip
Sample a vertex $u$ uniformly at random from $U$.

\smallskip
\uIf{$\deg{u} < \tau$}{
Choose a random vertex $v$ from $N(u) \cap U$ (if non-empty), add $(u,v)$ to $M$ and set $U \leftarrow U \setminus \set{u,v}$. If $N(u) \cap U = \emptyset$, set $U \leftarrow U \setminus \set{u}$.
}\Else{

\smallskip

Sample a vertex $v$ uniformly at random from $N(v)$.

\smallskip

\uIf{$v \in U$}{
add $(u,v)$ to $M$ and set $U \leftarrow U \setminus \set{u,v}$.
}
}
}

\textbf{Return} $M$.

\end{algorithm2e}

We remark that the first \textbf{if} condition in Algorithm~\ref{alg:mm} is used to remove the costly operation of checking if any unmatched edge is left in the graph. It is easy to see that
this algorithm always output an MM. 

We prove that Algorithm~\ref{alg:mm} is fast both in expectation and with high probability on graphs with bounded neighborhood independence. We also note that as stated, Algorithm~\ref{alg:mm} actually assumes knowledge of $\beta(G)$ (needed
for the definition of the threshold parameter $\tau$). However, we show at the end of this section that this assumption can be lifted easily and obtain a slight modification of Algorithm~\ref{alg:mm} with the same asymptotic runtime that does \emph{not} require any knowledge of $\beta(G)$.

\begin{theorem}\label{thm:mm-upper}
	The modified randomized greedy MM algorithm (as specified by Algorithm~\ref{alg:mm}) computes a maximal matching of a graph $G$ given in adjacency array representation in $O(n\log{n} \cdot \beta(G))$ time in expectation and with high probability.
\end{theorem}

	Let $t$ denote the number of iterations of the \textbf{while} loop in Algorithm~\ref{alg:mm}. We can bound the runtime of this algorithm based on $t$ as follows.
	
	\begin{claim}\label{clm:mm-runtime}
		Algorithm~\ref{alg:mm} can be implemented in $O(n \log{n} \cdot \beta(G) + t)$ time.
	\end{claim}
	\begin{proof}
First, we would like to store the set $U$ in a data structure that supports random sampling and deletion of a vertex from $U$, as well as determining whether a vertex is currently in $U$ or not, in constant time.
This data structure can be  easily implemented using two arrays $A_1$ and $A_2$; we provide the   rather tedious details for completeness. The arrays are initialized as $A_1[i] = A_2[i] = i$ for all $i = 1, \ldots, n$, where $A_2[i]$ holds the index of the cell in $A_1$ where $v_i$ is stored, or -1 if $v_i$ is not in $U$, while $A_1$ stores the 
vertices of $U$ in its first $|U|$ cells, identifying each $v_i$ with   index $i$.  %More accurately, while all cells of $A_2$ may be used at any point throughout the execution of the algorithm, the only relevant cells of $A_1$ are the first $|U|$ ones.
When any unmatched vertex $v_i$ is removed from $U$, we first use $A_2[i]$ to determine the cell where $v_i$ is stored in $A_1$, then we move the unmatched vertex  stored at the last cell in $A_1$, $A_1[|U|]$, to the cell currently occupied by $v_i$ by setting $A_1[A_2[i]] = A_1[|U|]$, and finally set  $A_2[A_1[|U|] = A_2[i],
A_2[i] = -1$. Randomly sampling a vertex from $U$ and determining whether a vertex belongs to $U$ can now be done in $O(1)$ time.  % using these arrays.
% $A_1[i] = A_1[|U|]$, $A_2[i] = A_
%We store $U$ throughout the algorithm in a simple array which 

		Using these arrays  each iteration of the {\bf while} loop in which $\deg{u} \geq \tau$ can be carried out within $O(1)$ time. 
Iterations for which $\deg{u} < \tau$ are more costly, due to the need to determine $N(u) \cap U$.
Nonetheless, in each such iteration we spend at most $O(\tau)$ time while at least one vertex is removed from $U$, hence the time required by all such iterations is bounded by $\sum_{k=1}^{n} O(\frac{n \cdot \beta(G)}{k}) = O(n\log{n} \cdot \beta(G))$.
It follows that the total runtime of the algorithm is $O(n \log{n} \cdot \beta(G) + t)$.
	\end{proof}
	
	The main ingredient of the analysis is thus to bound the number $t$ of iterations. Before proceeding we introduce some definition.
	We say that an iteration of the {\bf while} loop \emph{succeeds} iff we remove at least one vertex from $U$ in this iteration.
	Clearly, there can be at most $n$ successful iterations. We prove that each iteration of the
	algorithm is successful with a sufficiently large probability, using which we bound the total number of iterations.
	
	To bound the success probability, we shall argue that for sufficiently many vertices $u$ in $U$, the number of its neighbors in $U$, referred to as its \emph{internal} degree, is proportionate to the number of its neighbors
	outside $U$, referred to as its \emph{external} degree; for any such vertex $u$, a random neighbor $v$ of $u$ has a good chance of belonging to $U$.
	This is captured by the following definition.

	\begin{definition}[Good vertices]\label{def:delta-good}	
	For a parameter $\delta \in (0,1)$, we say a vertex $u \in U$ is \emph{$\delta$-good} iff $\degto{u}{U} \geq \delta \cdot \degto{u}{V \setminus U}$ or $\deg{u} < 1/\delta$.
	\end{definition}

	If we sample a $\delta$-good
	vertex $u \in U$ in some iteration, then with probability $\geq \delta/(1+\delta) \geq \delta/2$ that iteration succeeds. It thus remains to show that many vertices in $U$ are   good for an appropriate choice of $\delta$.
	We use the bounded neighborhood independence property (in a  more sophisticated way than it was used
	in the proof of Theorem~\ref{thm:mis-upper}) to prove the following lemma, which lies at the core of the analysis. See Figure~\ref{fig:mm-good} for an illustration.

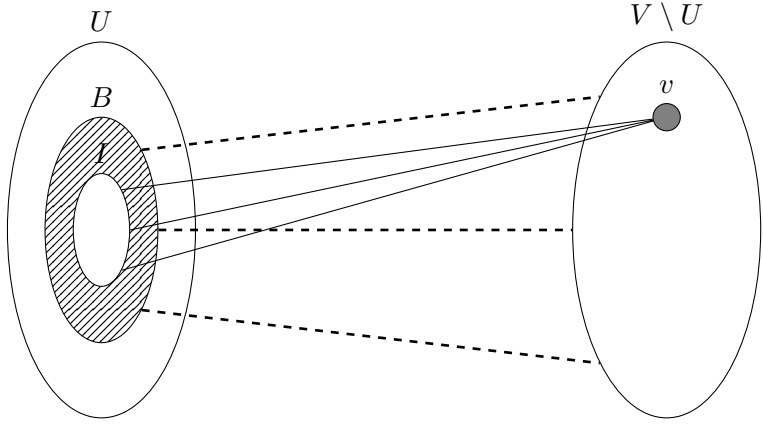
\begin{figure}[t!]
\centering
\begin{tikzpicture}
	\node[ellipse, draw, minimum height=5cm, minimum width=2.5cm, label=$U$] (U) {};
	\node[ellipse, draw, minimum height=5cm, minimum width=2.5cm, label=$V \setminus U$] (V) [right=5cm of U]{};
	\node[ellipse, draw, minimum height=3cm, minimum width=1.5cm, label=$B$, pattern=north east lines] (B) at ($(U)+(0pt,0pt)$){};
	\node[ellipse, draw, minimum height=1.5cm, minimum width=0.75cm, label=$I$, fill=white] (I) at ($(B)+(0pt,0pt)$){};
	\node[circle, draw, minimum height=0.05cm, minimum width=0.05cm, fill=gray, label=$v$] (v) at ($(V)+(0,1.5cm)$){};
	
	\draw[black, dashed, line width=1pt]
		(B.north east) -- (V.north west)
		(B.east) -- (V.west)
		(B.south east) -- (V.south west);
		
	\draw[black]
		(I.north east) -- (v)
		(I.east) -- (v)
		(I.south east) -- (v);
		
\end{tikzpicture}
\caption{An illustration of proof of Lemma~\ref{lem:mm-good}: $B$ is the set of vertices in $U$ that are \emph{not} $\delta$-good, i.e., have a ``large'' external degree. While $B$ is not necessarily an independent set, there exists 
a ``large'' independent set $I$ in $B$ (Claim~\ref{clm:mm-is-bad}). As vertices in $I$ have a large external degree, there must exists a single vertex $v \in V \setminus U$, incident on more than $\beta(G)$ vertices of $I$; a contradiction.}
\label{fig:mm-good}
\end{figure}

	 \begin{lemma}\label{lem:mm-good}
	 	Fix any choice of $U$ in some iteration and let $\delta := \delta(U) = 1/\tau(U)$ (for the parameter $\tau(U)$ defined in Algorithm~\ref{alg:mm}).
		Then at least half the vertices in $U$ are $\delta$-good in this iteration.
	 \end{lemma}
\begin{proof}
	Let us say that a vertex $u \in U$ is \emph{bad} iff it is not $\delta$-good (for the parameter $\delta$ in the lemma statement). Let $B$ denote the set of bad vertices and let $b := \card{B}$.
%	We have the following key claim.
%	The key step in the proof of this lemma is the following claim which is a generalization of Claim~\ref{clm:mis-alg-degree} in the proof of Theorem~\ref{thm:mis-upper}.

	\begin{claim}\label{clm:mm-is-bad}
		There exists an independent set $I \subseteq B$ with at least $\frac{b}{2\delta}$ edges leading to $V \setminus U$.
	\end{claim}
	\begin{proof}
		We prove this claim using a probabilistic argument. Pick a random permutation $\sigma$ of vertices in $B$ and add each vertex $v \in B$ to an initially empty independent set $I = I_\sigma$ iff $v$ appears before all its neighbors in $B$ according to $\sigma$. Clearly, the resulting set $I$ is an independent set inside $B$.
		
		Let $E(I)$ denote the set of edges that connect vertices of $I$ with vertices of $V \setminus U$.
		For any vertex $v \in B$, define a random variable $D_v \in \set{0,\degto{v}{V \setminus U}}$ which takes value equal to the external degree of $v$ iff $v$ is added to $I$. Clearly, $\card{E(I)} = \sum_{v \in V} D_v$.
		We have,
		\begin{align*}
			\Ex\card{E(I)} &= \sum_{v \in B} \expect{D_v} = \sum_{v \in B} \Pr\paren{v \in I} \cdot \degto{v}{V \setminus U} \\
			&= \sum_{v \in B} \frac{1}{\degto{v}{B}+1} \cdot \degto{v}{V \setminus U} \tag{$v$ is only chosen in $I$ iff it is ranked
			 first by $\sigma$ among itself and its $\degto{v}{B}$ neighbors} 
		\end{align*}
		\begin{align*}
			 &\qquad \geq \sum_{v \in B} \frac{1}{\delta \cdot \degto{v}{V \setminus U}+1} \cdot \degto{v}{V \setminus U} \tag{as $B \subseteq U$ and vertices in $B$ are all bad} \\
			 &\qquad \geq \sum_{v \in B} \frac{1}{2\delta \cdot \degto{v}{V \setminus U}} \cdot \degto{v}{V \setminus U} = \frac{b}{2\delta} \tag{as $\deg{v} \geq \frac{1}{\delta}$ and hence $\delta \cdot \degto{v}{V \setminus U} \geq 1$}.
		\end{align*}
		It follows that there exists a permutation $\sigma$ for which the corresponding independent set $I = I_\sigma \subseteq B$ has at least $\frac{b}{2\delta}$ edges leading to $V \setminus U$, finalizing the proof.
	\end{proof}
	
	We next prove Lemma~\ref{lem:mm-good} using an argument akin to Claim~\ref{clm:mis-alg-degree} in the proof of Theorem~\ref{thm:mis-upper}. Let $I$ be the independent set guaranteed by Claim~\ref{clm:mm-is-bad}.
	As at least $\frac{b}{2\delta}$ edges are going from $I$ to $V \setminus U$, a   double counting argument implies that there exists a vertex $v \in V \setminus U$ with degree to $I$ satisfying:
	\begin{align}
		\degto{v}{I} \geq \frac{b}{2\delta \cdot \card{V \setminus U}} \geq \frac{b\cdot\tau(U)}{2n}  = \frac{2b \cdot \beta(G)}{\card{U}}, \label{eq:mm-1}
	\end{align}
	where the  equality is by the choice of $\tau(U)$. Suppose towards a contradiction that $b > \card{U}/2$. This combined with Eq~(\ref{eq:mm-1}) implies that $\degto{v}{I} > \beta(G)$.
	Since $I$ is an independent set, it follows that $N(v)$ contains an independent set of size larger than $\beta(G)$, a contradiction. Hence $b \leq \card{U}/2$, and so at least half the vertices in $U$ are $\delta$-good, as required.
\end{proof}
	
We now use Lemma~\ref{lem:mm-good} to bound the expected number of iterations $t$ in Algorithm~\ref{alg:mm}. Let us define the random variables $X_1,\ldots,X_n$, where
$X_k$ denotes the number of iterations spent by the algorithm when $\card{U} = k$. Clearly, the total number of iterations $t = \sum_{k=1}^{n} X_k$.
We use these random variables to bound the expected value of $t$ in the following claim. The next claim then proves a concentration bound for $t$ to obtain the  high probability result.

\begin{claim}\label{clm:mm-t-expected}
	The number of iterations in Algorithm~\ref{alg:mm} is in expectation $\expect{t} \leq 16\beta(G) \cdot n\log{n}$.
\end{claim}
\begin{proof}
	As stated above, $\expect{t} = \sum_{k=1}^{n} \expect{X_k}$ and hence it suffices to bound each $\expect{X_k}$.
	Fix some $k \in [n]$ and consider the case when $\card{U} = k$. Recall the function $\delta(U)$ in Lemma~\ref{lem:mm-good}. As $\delta(U)$ is only a function of size of $k$,
	we slightly abuse the notation and write $\delta(k)$ instead of $\delta(U)$ where $k$ is the size of $U$.
	
	By Lemma~\ref{lem:mm-good},
	at least half the vertices in $U$ are $\delta(k)$-good.
	Hence, in each iteration, with probability at least half, we sample a $\delta(k)$-good vertex $u$ from $U$.
	Conditioned on this event, either $\deg{u} < 1/\delta(k)$ which means this iteration succeeds with probability $1$ or  $\degto{u}{U} \geq \delta(k) \cdot \degto{u}{V\setminus U}$, and hence
	with probability at least $\delta(k)/(1+\delta(k)) \geq \delta(k)/2$, the sampled vertex $v$ belongs to $U$ and again this iteration succeeds. As a result, as long as $U$ has not changed, each iteration has probability at least $\delta(k)/4$ to succeed.
	
	By the above argument, $X_k$ statistically dominates a Poisson distribution with parameter $\delta(k)/4$ and hence $\expect{X_k} \leq 4/\delta(k)$.
	To conclude,
	 \begin{align*}
	 	\expect{t} = \sum_{k=1}^{n} \expect{X_k} \leq \sum_{k=1}^{n} 4/\delta(k) = 16\beta(G) \cdot n \cdot  \sum_{k=1}^{n} \frac{1}{k} \leq 16\beta(G) \cdot n\log{n}.
	 \end{align*}
	 which finalizes the proof.
\end{proof}

Claim~\ref{clm:mm-t-expected} combined with Claim~\ref{clm:mm-runtime} is already enough to prove the expected runtime bound in Theorem~\ref{thm:mm-upper}.
We now prove a concentration bound for $t$ to obtain the high probability bound. We note that it seems possible to prove the following claim by using 
standard concentration inequalities; however doing so requires taking care of several boundary cases for the case when $\card{U}$ become $o(\log{n})$, and hence 
we instead prefer to use the following direct and more transparent proof. 

%%We shall note that this bound, while not hard to prove, does not seem to follow directly from standard concentration bounds such as Chernoff-Hoeffding or Azuma's inequality (but can be obtained trivially
%%if we allow for an additional factor of $O(\log{n})$ which we are saving here).
%% \shay{similar changes in proof of next claim, re dividing or multiplying by $1+\delta(k)$; maybe we should define another parameter say $\zeta(k)$, which is just $\delta(k)/(1+\delta(k))$, and then work with $\zeta(k)$?}
\begin{claim}\label{clm:mm-t-whp}
	The number of iterations in Algorithm~\ref{alg:mm} is $t = O(\beta(G) \cdot n\log{n})$ with high probability.
\end{claim}
\begin{proof}
	Recall from the proof of Claim~\ref{clm:mm-t-expected} that $t = \sum_{k=1}^{n} X_k$ and that each $X_k$ statistically dominates a Poisson distribution with parameter $\delta(k)/4$ (as defined in Claim~\ref{clm:mm-t-expected}).
	Define $Y_1,\ldots,Y_n$
	as independent random variables where $Y_k$ is distributed according to exponential distribution with mean $\mu_k := \frac{4}{\delta(k)}$.
	For any $x \in \IR^+$,
	\[
		\Pr\paren{X_k \geq x} \leq \paren{\frac{\delta(k)}{4}}^{x} = \Pr\paren{Y_k \geq x}.
	\]
	
	As such, the random variable $Y := \sum_{k=1}^{n} Y_k$ statistically dominates the random variable $t$ for number of iterations. Moreover by Claim~\ref{clm:mm-t-expected},
	$\mu := \expect{Y} = 16\beta(G) \cdot n\log{n}$ (the equality for $Y$ follows directly from the proof).
	
	In the following, we prove that with high probability $Y$ does not deviate from its expectation by much. The proof follows the standard moment generating function idea (used for instance in the proof of Chernoff-Hoeffding bound).
	Let $y \in \IR^+$. For any $s > 0$,
	\begin{align}
		\Pr\paren{Y \geq y} = \Pr\paren{e^{sY} \geq e^{s y}} \leq \frac{\expect{e^{sY}}}{e^{s y}},  \label{eq:mm-markov}
	\end{align}
	where the inequality is simply by Markov bound. Additionally, since $Y = \sum_{k=1}^{n} Y_k$ and and $Y_k$'s are independent, we have for any $s > 0$,
	\begin{align}
		\frac{\expect{e^{sY}}}{e^{s y}} = \frac{\expect{e^{s \cdot \sum_{k=1}^{n} Y_k}}}{e^{s y}} = \frac{\prod_{k=1}^{n} \expect{e^{s \cdot Y_k}}}{e^{s y}}. \label{eq:mm-prod}
	\end{align}
	Recall that for every $i \in [n]$, $\expect{Y_k} = \mu_k$ and $Y_k$ is distributed according to exponential distribution. Thus, for any $s < 1/\mu_k$,
	\begin{align}
		\expect{e^{sY_k}} = \int_{y=0}^{\infty} e^{s y} \cdot \Pr\paren{Y_k = y} dy = \frac{1}{\mu_k} \int_{y=0}^{\infty} e^{s y} \cdot e^{-y/\mu_k} dy = \frac{1}{1-s \cdot \mu_k}. \label{eq:mm-int}
	\end{align}
	 Recall that $\mu_k = {4}/{\delta(k)}$ for every $k \in [n]$ and $\delta(1) < \delta(2) < \ldots < \delta(n)$ by definition.
	 Pick $s^* = {1}/{2\mu_1}$ and so $s^* < 1/\mu_k$ for all $k \in [n]$. By plugging in the bounds in Eq~(\ref{eq:mm-int}) for $s=s^*$ into Eq~(\ref{eq:mm-prod}),
	 we have,
	 \begin{align*}
	 	\frac{\expect{e^{s^*Y}}}{e^{s^* y}} &= \frac{\prod_{k=1}^{n} \expect{e^{s^* \cdot Y_k}}}{e^{s^* y}} = \prod_{k=1}^{n}\frac{e^{-s^* y}}{1-s^*\mu_k} \leq e^{-s^* y} \cdot \exp\Paren{2\sum_{k=1}^{n} s^*\mu_k} \tag{as $1-x \geq e^{-2x}$ for $x \in (0,1/2]$} \\
		&= \exp\Paren{-s^* y + 2 s^*\mu} = \exp\paren{-y/2\mu_1 + \mu/\mu_1}.  \\
	 \end{align*}
	 We now plug in this bound into Eq~(\ref{eq:mm-markov}) with the choice of $y = 4\mu$ to obtain that,
	 \begin{align*}
	 	\Pr\paren{Y \geq 4\mu} \leq \exp\paren{-4\mu/2\mu_1 + \mu/\mu_1} = \exp\paren{-\mu/\mu_1} = \exp{\paren{-\log{n}}} = 1/n,
	 \end{align*}
	 where we used the fact that $\mu/\mu_1 \geq \log{n}$. This means that with high probability, $Y$ is only $4$ times larger than its expectation, finalizing the proof.
\end{proof}

The high probability bound in Theorem~\ref{thm:mm-upper} now follows from Claim~\ref{clm:mm-runtime} and Claim~\ref{clm:mm-t-whp}, concluding the whole proof of this theorem.

\subsection*{Unknown $\beta(G)$}

We next show that our algorithm can be easily adjusted to the case when $\beta(G)$ is unknown. The idea is simply to ``guess'' $\beta(G)$ in powers of two,   starting from $\beta = 2$ and ending at $\beta=n$, and each time to (sequentially) run 
Algorithm~\ref{alg:mm} under the assumption that $\beta(G) = \beta$. For each choice of $\beta$, we shall only run the algorithm for at most $t=O(n\log{n} \cdot \beta)$ iterations (where the constant hiding in the $O$-notation should be sufficiently large, in 
accordance with that in the proof of Claim~\ref{clm:mm-t-whp}) and if at the end of a run the set $U$ in the algorithm has not become empty, we   start a new run from scratch with the next (doubled) value of $\beta$. (For $\beta=n$, we do not terminate the 
algorithm prematurely and instead run it until $U$ is empty.)

By Theorem~\ref{thm:mm-upper}, for the \emph{first} choice of $\beta$ for which $\beta \geq \beta(G)$, the algorithm must terminate with high probability within $O(n\log{n} \cdot \beta(G))$ time (as $\beta \leq 2\beta(G)$ also). 
Moreover, the runtime of every previous run is bounded {\em deterministically} by $O(n\log{n} \cdot \beta)$ (for the corresponding guess $\beta$ of $\beta(G)$).
Consequently, the total runtime is
\begin{align*}
	O(n\log{n}) \cdot \sum_{\set{2^{i} \mid 2^{i} \leq 2\beta(G)}} 2^{i} = O(n\log{n} \cdot \beta(G)),
\end{align*}
where this bound holds with high probability.
In this way we get an algorithm that uses no prior knowledge of $\beta(G)$ and
achieves the same asymptotic performance as Algorithm~\ref{alg:mm}.

\renewcommand{\bM}{\ensuremath{\overline{M}}}

\renewcommand{\bM}{\ensuremath{\overline{M}}}

\newcommand{\Cused}{\ensuremath{C_{\textnormal{\textsf{used}}}}}
\newcommand{\Cfree}{\ensuremath{C_{\textnormal{\textsf{free}}}}}

\section{A Lower Bound for Deterministic Maximal Matching}\label{sec:mm-lower}

We  prove that randomization is necessary to obtain a sublinear time algorithm for MM even on graphs with bounded neighborhood independence.

\begin{theorem}\label{thm:mm-lower}
	Any deterministic algorithm that finds a maximal matching in every given graph $G$ with neighborhood independence $\beta(G) = 2$ (known to the algorithm) presented in the adjacency array representation requires $\Omega(n^2)$ time.
\end{theorem}

%%The proof of this theorem is based on setting up an ``evasive'' game between any given deterministic algorithm for MM and an adversary that answers the probes of the algorithm to the adjacency array of the underlying graph. 
%%We show that there exists a family of graphs with neighborhood independence number $2$ (in particular cliques minus a perfect matching) and a strategy for the adversary to answer the probes of the algorithm in such a way that $(a)$ the answers 
%%are consistent with at least one graph in this family and $(b)$ unless the algorithm makes $\Omega(n^2)$ queries to the adversary, the resulting matching it outputs cannot be maximal. This then gives us the proof of Theorem~\ref{thm:mm-lower}.
%%Due to space limitations, the formal proof is postponed to Appendix~\ref{app:mm-lower}. 
%%

For every integer $n=10k$ for $k \in \IN$, we define $\FG_n$ as the family of all graphs obtained by removing a perfect matching from a clique $K_n$ on $n$ vertices.
For a graph $G$ in $\FG_n$ we refer to the removed perfect matching of size $5k$ as the \emph{non-edge matching} of $G$ and denote it by $\bM(G)$. Clearly, every graph
in $\FG_n$ has neighborhood independence $\beta(G) = 2$. Moreover, any MM in $G$ can have at most $2$ unmatched vertices.

Let $\alg$ be a deterministic algorithm for computing an MM on every graph in $\FG_n$. We prove Theorem~\ref{thm:mm-lower} by analyzing a game between $\alg$ and an adaptive adversary that
answers the probes of $\alg$ to the adjacency array of the graph. In particular, whenever $\alg$ probes a \emph{new} entry of the adjacency array for some vertex $v \in V$, we can think of $\alg$ making a \emph{query} $Q(v)$ to
the adversary, and the adversary outputs a vertex $v$ that had not been so far revealed in the neighborhood of $u$ (as degree of all vertices in $G$ is exactly $n-2$, $\alg$ knows the degree of all vertices and thus does not need to make 
any degree queries at all). 

We now show that there is a strategy for the adversary to answer the queries of $\alg$ in a way that ensures $\alg$ needs to make $\Omega(n^2)$ queries before it can output an MM of the graph.

\begin{figure}[t!]
\centering
\begin{subfigure}{.45\textwidth}
\begin{tikzpicture}
	\node[ellipse, draw, minimum height=2cm, minimum width=1cm, label=$D$, pattern=north east lines] (D) {};
	\node[ellipse, draw, minimum height=3cm, minimum width=1.5cm, label={[xshift=1.5cm, yshift=-1.75cm]$\Cfree$}] (Cfree) [above right=-0.4cm and 2cm of D]{$?$};
	\node[ellipse, draw, minimum height=2cm, minimum width=1cm, label={[xshift=1.5cm, yshift=-1.25cm]$\Cused$}, pattern=north east lines] (Cused) [below=0.5cm of Cfree]{};
	
	 \tikzset{decoration={snake,amplitude=.4mm,segment length=2mm,
                       post length=0mm,pre length=0mm}}
                       
	\draw[black, decorate, line width=1pt]
		(D.north east) -- (Cfree.west)
		(D.south east) -- (Cused.west)
		(Cused.north) -- (Cfree.south);

\end{tikzpicture}
\caption{Algorithm's ``knowledge''  after all queries.}
\end{subfigure}
~
\begin{subfigure}{.45\textwidth}
\begin{tikzpicture}

	\node[ellipse, draw, minimum height=2cm, minimum width=1cm, label=$D$, pattern=north east lines] (D) {};
	\node[ellipse, draw, minimum height=3cm, minimum width=1.5cm, label={[xshift=1.5cm, yshift=-1.75cm]$\Cfree$}] (Cfree) [above right=-0.4cm and 2cm of D]{};
	\node[ellipse, draw, minimum height=2cm, minimum width=1cm, label={[xshift=1.5cm, yshift=-1.25cm]$\Cused$}, pattern=north east lines] (Cused) [below=0.5cm of Cfree]{};
	
	 \tikzset{decoration={snake,amplitude=.4mm,segment length=2mm,
                       post length=0mm,pre length=0mm}}
                       
%%	\draw[black, decorate, line width=1pt]
%%	%	(D.north east) -- (Cfree.north west)
%%		(D.east) -- (Cfree.west)
%%		(Cused.north) -- (Cfree.south);
		
%%	\node[ellipse, draw, minimum height=3cm, minimum width=1.5cm, label=$D$] (D) {};
%%	\node[ellipse, draw, minimum height=5cm, minimum width=2.5cm, label=$C$] (C) [right=2cm of D]{};
%%	
%%	\node (Cfree) at ($(C)+(0cm,0.75cm)$) {$\Cfree$};
%%	\node (Cused) at ($(C)+(0cm,-1.5cm)$) {$\Cused$};
%%	
%%	\draw[black]
%%		($(C)+(-1.2cm,-0.75cm)$) -- ($(C)+(1.2cm,-0.75cm)$); 
%%		
	\draw[line width=0.5pt]
		($(D) + (0.2cm,-0.15cm)$) -- ($(Cfree)+(-0.5cm,0.25cm)$)
		($(D) + (0.2cm,0.35cm)$) -- ($(Cfree)+(-0.5cm,0.75cm)$)
		($(Cused) + (-0.2cm,0.35cm)$) -- ($(Cfree) + (-0.2cm,-0.85cm)$)
		($(Cused) + (0.2cm,0.35cm)$) -- ($(Cfree) + (0.2cm,-0.85cm)$);
	
	\draw[line width=0.5pt]
		($(Cfree) + (-0.45cm,-0.2cm)$) -- ($(Cfree) + (0.45cm,-0.2cm)$) node[above left=-1pt and 6pt]{$\mathbf{?}$};	
		
%%	
%%	\draw[blue]
%%		($(C) + (-0.75cm,-1.25cm)$) -- ($(C)+(-0.75cm,-0.15cm)$)
%%		($(C) + (-0.25,-1.25cm)$) -- ($(C)+(-0.25,-0.15cm)$)
%%		($(C) + (0.25cm,-1.25cm)$) -- ($(C)+(0.25cm,-0.15cm)$)
%%		($(C) + (0.75cm,-1.25cm)$) -- ($(C)+(0.75cm,-0.15cm)$);
%%		
%%	\draw[red ,line width=1pt]
%%		($(C) + (-0.25cm,1.75cm)$) -- ($(C) + (0.75cm,1.75cm)$)
%%		($(C) + (-0.25cm,1.25cm)$) -- ($(C) + (0.75cm,1.25cm)$) node[above left=-1pt and 6pt]{$\mathbf{?}$};
%%		
			
%%	\node[ellipse, draw, minimum height=3cm, minimum width=1.5cm, label=$B$, pattern=north east lines] (B) at ($(U)+(0pt,0pt)$){};
%%	\node[ellipse, draw, minimum height=1.5cm, minimum width=0.75cm, label=$I$, fill=white] (I) at ($(B)+(0pt,0pt)$){};
%%	\node[circle, draw, minimum height=0.05cm, minimum width=0.05cm, fill=gray, label=$v$] (v) at ($(V)+(0,1.5cm)$){};
%%	
%%	\draw[black, dashed, line width=1pt]
%%		(B.north east) -- (V.north west)
%%		(B.east) -- (V.west)
%%		(B.south east) -- (V.south west);
%%		
%%	\draw[black]
%%		(I.north east) -- (v)
%%		(I.east) -- (v)
%%		(I.south east) -- (v);
		
\end{tikzpicture}
\caption{The output maximal matching.}
\end{subfigure}
\caption{An illustration of proof of Theorem~\ref{thm:mm-lower}: After making ``small'' number of queries, the algorithm can know all edges between dummy vertices and vertices in $\Cused$, but no edge inside $\Cfree$ has been discovered. On the other hand, the algorithm should also output a matching that contains some edges with both endpoints inside $\Cfree$ which leads to a contradiction.}
\label{fig:mm-lower}
\end{figure}
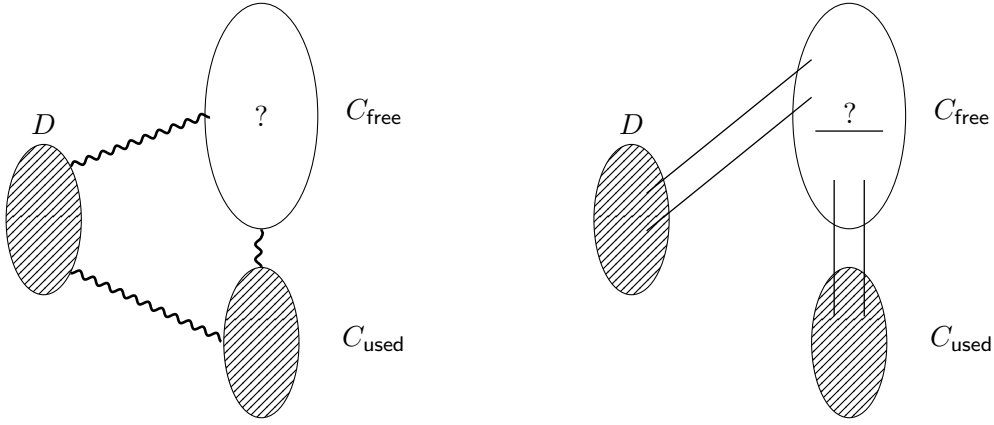

\paragraph{Adversary's Strategy:} The adversary picks an arbitrary set of $2k$ vertices $D$ referred to as \emph{dummy} vertices. We refer to remaining vertices as \emph{core} vertices and denote them by $C:= V \setminus D$.
The adversary also fixes a non-edge matching of size $k$ between vertices in $D$, denoted by $\bM_D$.
The non-edge matching of $G$ consists of $\bM_D$ and a non-edge matching of size $4k$ between vertices in $C$, denoted by $\bM_C$, which unlike $\bM_D$ is constructed adaptively by the adversary.
We assume $\alg$ knows the partitioning of $V$ into $D$ and $C$, as well as the non-edge matching $\bM_D$. Hence, the only missing information to $\alg$ is identity of $\bM_C$.

To answer a query $Q(u)$ for a dummy vertex $u \in D$, the adversary simply returns any arbitrary vertex in $V$ (not returned so far as an answer to $Q(u)$) except for the mate of $u$ in $\bM_D$, which cannot be a neighbor of $u$.
To answer the queries $Q(w)$ for core vertices $w \in C$, the adversary maintains a partitioning of $C$ into $\Cused$ and $\Cfree$. Initially all core vertices belong to $\Cfree$ and $\Cused$ is empty.
Throughout we only move vertices from $\Cfree$ to $\Cused$. The adversary also maintains a counter for every vertex in $\Cfree$ on how many times that vertex has been queried so far.
Whenever a vertex $w \in \Cfree$ is queried, as long as this vertex has been queried at most $2k$ times, the adversary returns an arbitrary dummy vertex $u$ from $D$ as the answer to $Q(w)$ (which is possible because of size of $D$ is $2k$).
Once a vertex $w \in \Cfree$ is queried for its $(2k+1)$-th time, we pick another vertex $w'$ from $\Cfree$ also, add the pair $(w,w')$ to the non-edge matching $\bM_C$ and move both $w$ and $w'$ to $\Cused$,
and then answer $Q(w)$ for the case $w \in \Cused$ as described below.

Recall that for any vertex $w \in \Cused$, by construction, there is another fixed vertex $w'$ in $\Cused$ (joined at the same time with $w$) where $(w,w') \in \bM_C$. For any query for $w \in \Cused$, the adversary answers $Q(w)$ by returning an arbitrary vertex
from $C \setminus \set{w'}$. This concludes the description of the strategy of the adversary.

\medskip
We have the following basic claim regarding the correctness of the adversary's strategy. %The proof is straightforward.

\begin{claim}\label{clm:adversary-correct}
	The answers returned by the adversary for any sequence of queries are always consistent with at least one graph $G$ in $\FG_n$.
\end{claim}
\begin{proof}
	We can append to the current sequence of queries a sufficiently long sequence that ensures all vertices are queried  $n-2$ times. Thus,
	the adversary would eventually construct the whole non-edge matching $\bM_C$ also. There exists a unique graph $G$ in $\FG_n$ where $\bM(G) = \bM_D \cup \bM_C$, hence proving the claim (note that before appending the sequence,
	there may be multiple graphs consistent with the original sequence).
\end{proof}

We now prove the following lemma, which is the key step in the proof of Theorem~\ref{thm:mm-lower}.
\begin{lemma}\label{lem:mm-contradiction}
	Suppose $\alg$ makes at most $2k^2$ queries to the adversary and outputs a matching $M$ using only these queries. Then, there exists some graph $G$ in $\FG_n$ where $G$ is consistent with the answers returned by the adversary to $\alg$
	and $\bM(G) \cap M \neq \emptyset$.
\end{lemma}
\begin{proof}
	Since $\alg$ makes at most $2k^2$ queries, there can only be $2k$ vertices in $\Cused$ by the time the algorithm finishes its queries (as each pair of vertices in $\Cused$ consume $2k$ queries at least).
	
	Consider the maximal matching $M$. There are at most $2k$ edges of $M$ that are incident on $\Cused$ and $2k$ more edges incident on $D$. This implies that at most $4k$ vertices in $\Cfree$ are matched to vertices
	outside $\Cfree$. As $n=10k$, we have $\card{\Cfree} \geq 6k$, which means that there are at least $2k$ vertices in $\Cfree$ that are not matched by $M$ to vertices \emph{outside} $\Cfree$. 
	
	Now, the maximality of $M$, together with the fact that $G \in \FG_n$, ensures that there is an edge $(u,v)$ in $M$ with both endpoints in $\Cfree$ (in fact, there are $\Omega(k)$ such edges). However, note that the adversary has not 
	committed to the non-edge matching $\bM_C$ inside $\Cfree$ yet, and in particular, can make sure that $(u,v)$ belongs to $\bM_C$. But this is a contradiction with the correctness of the algorithm as it outputs an edge that does not belong to the graph. 
\end{proof}

Theorem~\ref{thm:mm-lower} now follows immediately from Lemma~\ref{lem:mm-contradiction}, which states that unless the algorithm makes $\Omega(n^2)$ queries, there always exists at least one graph in $\FG_n$ for which the output matching of the algorithm 
is not feasible. As graphs in $\FG_n$ all have $\beta(G) = 2$, we obtain the final result.

%%\subsection*{Acknowledgements}

\bibliographystyle{abbrv}
\bibliography{general}

%%\clearpage
%%\appendix

\end{document}